\documentclass[12pt]{article}
\usepackage{cite}
\usepackage{tikz}
\usepackage{amsmath, amsthm, amssymb}
\newtheorem{thm}{Theorem}[section]

\newtheorem{lem}[thm]{Lemma}
\newtheorem{ex}[thm]{Example}
\newtheorem{defn}[thm]{Definition}
\newcommand{\ie}{\textit{i.e.}\ }
\newcommand{\eg}{\textit{e.g.}\ }
\begin{document}

\author{John Faben\footnote{School of Mathematical Sciences, Queen Mary, University of London, Mile End Road, London, E1 4NS, UK}}

\title{The complexity of counting solutions to Generalised Satisfiability Problems modulo $k$ \footnote{Supported by EPSRC grant No.EP/E064906/1, `The Complexity of Counting in Constraint Satisfaction Problems'}}
\maketitle

\begin{abstract}
Generalised Satisfiability Problems (or Boolean Constraint Satisfaction Problems), introduced by Schaefer in 1978, are a general class of problem which allow the systematic study of the complexity of satisfiability problems with different types of constraints. In 1979, Valiant introduced the complexity class parity P, the problem of counting the number of solutions to NP problems modulo two. Others have since considered the question of counting modulo other integers.

We give a dichotomy theorem for the complexity of counting the number of solutions to Generalised Satisfiability Problems modulo integers. This follows from an earlier result of Creignou and Hermann which gave a counting dichotomy for these types of problem, and the dichotomy itself is almost identical. Specifically, counting the number of solutions to a Generalised Satisfiability Problem can be done in polynomial time if all the relations are affine. Otherwise, except for one special case with $k = 2$, it is \#$_k$P-complete.
\end{abstract}
\section{Introduction} \label{SectionIntro}

The complexity class $\bigoplus $P (pronounced `parity P') was first introduced by Valiant in \cite{Val78}. It formalises the question of counting the parity of the number of solutions to NP problems. Formally, it is the class of languages $S$ such that there is a polynomial time Turing Machine which on input $x \in S$ has an odd number of accepting computations and on $x \not \in S$ has an even number. In this paper, we will also deal with the problem of counting the  number of solutions to NP problems modulo other integers. We will give details of the complexity classes used to deal with this problems in section \ref{SectionClasses}. There have been some interesting recent results in this area, with Valiant proving that there exist problems which are complete for $\bigoplus$P but for which counting the number of solutions \textit{modulo 7} can be done in polynomial time \cite{Val06}.

Generalised Satisfiability Problems (also referred to as Boolean Constraint Satisfaction Problems) are a very general class of problem, which provide the base cases for the reductions in a wide variety of complexity theoretic proofs. They were first studied by Schaefer in \cite{Sch78}, where he proved a dichotomy theorem for the decision version of these problems (assuming P$\not=$NP). The Generalised Satisfiability Problem is this: given a set $S$ of boolean relations, the $S$-satisfiability problem is the question of determining whether or not a given $S$-formula is satisfiable, where an $S$-formula is a conjunction of S-relations. The set of all satisfiable $S$-formulae is denoted by SAT($S$). For example, if $S$ were the set of all eight 3-ary boolean relations, SAT($S$) would be the well-known 3-SAT language. Schaefer showed that the decision versions of Generalised Satisfiability Problems can be divided into two classes - those which are NP-complete, and those which can be solved in polynomial time, depending on what type of logical relations is contained in the set $S$. This is in contrast with a result of Ladner that, under the assumption $\mathrm{P \not= NP}$, there is an infinite hierarchy of problems of increasing complexity between problems in P and problems which are NP complete\cite{Lad75}.

A dichotomy theorem for the counting version was proved by Creignou and Hermann in \cite{CreHer96}. They show that the counting version of a Generalised Satisfiability Problem $\#\mathrm{SAT}(S)$ can be solved in polynomial time if the all the relations in $S$ are affine; if not, $\#\mathrm{SAT}(S)$ is $\#P$-complete. A revised version of their proof appears in the monograph \cite{CreKha01}, results from which are used in section \ref{SectionSAT}.

Given this counting dichotomy, we are motivated to pose the question: among those Satisfiability problems for which the counting problem is known to be \#P-complete are there any for which the number of solutions is easy to count modulo some integer $k$? The answer is almost always no. The dichotomy we find in this paper is identical to that found in \cite{CreKha01} except for one difference for the case $k=2$.

\section{The classes \#$_k$P} \label{SectionClasses}

Previous work dealing with the complexity of counting modulo integers (\eg \cite{CaiHem89} \cite{Hert90}) has tended to define the relevant complexity class as Mod$_k$P, the set of languages which have non-zero number of accepting paths modulo $k$ for some Turing Machine $M$. Formally, Mod$_k$P contains for every function $f \in \#P$ the language
\begin{equation*}\{x \in \Sigma^*  \mid  f(x) \not\equiv   0  \pmod k\} \end{equation*}

 For the purposes of the work in this paper, we have chosen to define a slightly different set of classes, which we refer to here as $\#_k$P, and which we think more intuitively capture the notion of counting modulo $k$. Analogous to $\#P$, we define $\#_k$P to be class of problems ``compute $f(x)$ modulo $k$" where $f(x)$ is the number of accepting paths of a polynomial time Turing Machine. Like \#P, this is a class of function problems rather than a class of decision problems. Formally:

\begin{defn} \label{DefnSharpkP}
Let $\#acc_{M_k}$ be the function mapping from an input $x$ to the number of accepting paths of the non-deterministic Turing Machine $M$ on input $x$ modulo $k$. The class $\#_k$P consists of all functions $\#acc_{M_k}$ for all non-deterministic Turing Machines $M$ with polynomial length accepting paths on input $x$.
\end{defn}

 It seems intuitively that there should be problems for which determining the number of solutions modulo $k$ exactly is harder than determining whether the number of solutions modulo $k$ is non-zero. We have been able to construct artificial examples of such problems, but whether any natural problems with this property exist is an open question.

It should be noted that previous papers have used both $\#_k$P and Mod$_k$P to refer to the decision class defined above as Mod$_k$P.

We will need use the notion of a reduction which is parsimonious modulo $k$; just as a parsimonious reduction from one counting problem to another is one which preserves the number of solutions exactly, so a reduction which is parsimonious modulo $k$ is one which preserves exactly the number of solutions modulo $k$. We note in passing that a reduction which is parsimonious is also parsimonious modulo $k$ for all $k$.

\begin{defn} \label{DefnParsimonious}
Given two counting problems $\#A$ and $\#B$, we say there is a parsimonious reduction from $\#A$ to $\#B$ if there exists a function $f$ computable in polynomial time such that for all $x$, $|\{y:(x,y)\in A\}|= \{z:(f(x),z) \in B\}|$.
\end{defn}
\begin{defn} \label{DefnParsimoniousModk}
Given two \#$_k$P counting problems, $\#A$ and $\#B$, we say there is a parsimonious reduction modulo $k$ (a $\#_k$-reduction) from $\#A$ to $\#B$ if there exists a polynomially computable function $f$ such that $|\{y:(x,y)\in A\}|\equiv |\{z:(f(x),z) \in B\}| \pmod k$. In this case we say $\#A \leq_{\#_k} \#B$.
\end{defn}

Again, in an analogy with \#P completeness, we define the notion of \#$_k$P-completeness to with respect to Turing reducibility. Essentially, a problem \#$_k$A is \#$_k$P-complete if every problem in \#$_k$P can be solved in polynomial time given an oracle for \#$_k$A.

\begin{defn} \label{DefnCompleteness}
A counting problem $\#_kA \in \#_k$P is \#$_k$P-complete if for all other problems $\#_kB \in \#_k$P, $\#_kB$ can be solved in polynomial time with a $\#_k$P oracle for $\#_k A$.
\end{defn}

\section{Preliminaries} \label{SectionPrelims}

In \cite{CreKha01}, the counting dichotomy for Generalised Satisfiability Problems is established via reductions to problems referred to in that paper as \#SAT$(\mathrm{OR}_0)$, \#SAT$(\mathrm{OR}_1)$ and \#SAT$(\mathrm{OR}_2)$. These are the problems of counting the number of satisfying assignments of boolean formulae whose constraints are defined by functions of the form $x_i\vee x_j$, $\bar{x_i} \vee x_j$ and $\bar{x_i} \vee \bar{x_j}$ respectively. In this paper, we will use essentially the same reductions to find a dichotomy for counting modulo $k$ for all integer $k$. We therefore begin by proving the following \#$_k$-hardness result.

\begin{thm} \label{ORComplexity}
The problems \#$_k$SAT$(\mathrm{OR}_0)$, \#$_k$SAT$(\mathrm{OR}_1)$ and \#$_k$SAT$(\mathrm{OR}_2)$ are \#$_k$P-complete for all $k$.
\end{thm}

The proof of this theorem will be in several stages, and will be by reduction of the satisfiability problems to counting the number of independent sets in various classes of graph. Specifically, to \#$_k$INDEPENDENT-SET, the problem of counting the number of independent sets in a general graph modulo $k$ and to \#$_k$BIPARTITE-INDEPENDENT-SET, the problem of counting the number of independent sets in a bipartite graph modulo $k$.

We begin by noting that \#$_k$SAT$(\mathrm{OR}_2)$ is trivially reducible to \#$_k$SAT$(\mathrm{OR}_0)$, simply by taking the negation of each literal in the formula. We then make use of the following lemma, which simply states a well-known equivalence between the two problems:
\begin{lem}\label{Mon2SATisIS}
\#$_k$INDEPENDENT-SET $\leq_{\#_k}$ \#$_k$SAT$(\mathrm{OR}_2)$ for all $k$.
\end{lem}
\begin{proof}

With a graph $G$ on vertices $\{v_1,\ldots,v_n\}$ we associate the $\mathrm{OR}_2$ formula $F$ on the variables $\{x_1, \ldots,x_n\}$ such that the clause $\bar{x_i}\vee \bar{x_j}$ appears in $F$ if and only if there is an edge between vertices $v_i$ and $v_j$ in $G$. Then given an independent set, $I$ in $G$, the truth assignment, $s$, which satisfies $s(x_i) = \textrm{true} \iff v_i \in I$ is a satisfying assignment for $F$, and vice versa, given a satisfying assignment, the corresponding vertex set is independent. So the satisfying assignments of $F$ are in one-to-one correspondence with the independent sets of $G$, and the reduction is parsimonious.

Since the reduction given above is parsimonious, it is parsimonious modulo $k$ for all $k$.
\end{proof}

We will also make use of the following lemma.

\begin{lem}\label{OR1BiIS}
\#$_k$BIPARTITE-INDEPENDENT-SET $\leq_{\#_k}$ \#$_k$SAT$(\mathrm{OR_1})$ for all $k$.
\end{lem}

\begin{proof}
The reduction given by Linial in \cite{lin86} preserves the number of solutions exactly, and so preserves the number of solutions modulo $k$ for all $k$.
\end{proof}

In the rest of this section, we will show that the problems, \#$_k$INDEPENDENT-SET and \#$_k$BIPARTITE-INDEPENDENT-SET are both \#$_k$-complete for all integer $k$. This will be done by reduction from the problem \#$_k$SAT, which we define as the problem of counting the number of satisfying assignments of a boolean formula in constructive normal form modulo $k$. This problem is known to be \#$_k$-complete since the reduction used in Cook's Theorem can be made parsimonious \cite{Sim75}.

The reductions used in the proofs which follow all have the same basic structure. Given a SAT-formula, we produce a graph in which the independent sets with a certain property all correspond to satisfying assignments of $F$, and in which the independent sets which do not have this property can be partitioned into $k$ subsets of equal size, the total number of which is therefore zero modulo $k$. This allows us to produce a formula for the number of independent sets modulo $k$, as described in lemma \ref{ISFormula}.

In the following, $\mathcal{I}(G)$, where $G$ is a graph will denote the set of independent sets of $G$. We will also use $\mathcal{I}(G;X)$, where $X$ is a set of vertices in a graph $G$ to denote the set of independent sets of $G$ all of whose vertices lie in $X$. We will also use $N(x)$ to represent the (open) neighbourhood in $G$ of a vertex $x$ and $N_G(H)$ to represent the (open) neighbourhood of a subset, $X \subset V(G)$.

\begin{lem}\label{ISFormula}
Consider a graph $G$ with the following structure:

$G$ consists of a set of vertices $X$, along with $n$ copies of a graph $H$, $\{H_1,\ldots, H_n\}$, each of which contains distinguished vertex $h_i$. The edges in $G$ either go between vertices in one copy of $H$, between vertices in $X$ or between some distinguished vertex $h_i$ and a vertex in $X$. Furthermore, $H$ has the property that the total number of independent sets in $H$ is congruent to zero modulo $k$.

The total number of independent sets in $G$ is congruent modulo $k$ to:
\begin{equation}
\displaystyle\sum_{I_0 \in \mathcal{I}(G;X)}\prod_{i=1}^n \min\{|I_0\cap N_G(H_i)|,1\} \times |\mathcal{I}(G;H_i\backslash\{h_i\})|
\end{equation}

\end{lem}
\begin{proof}
The relevant intuition for this proof is that if we have two sets of vertices, say $X$ and $Y$, satisfying $N_G(X)\cap Y = \emptyset$ then $|\mathcal{I}(G;X\cup Y)| = |\mathcal{I}(G;X)| \times |\mathcal{I}(G;Y)|$. This is because any independent set which lies entirely in $X\cup Y$ is the union of an independent set in $X$ and an independent set in $Y$, and each such union in an independent set by the condition on the neighbourhoods.

We note that if $J \in \mathcal{I}(G)$ is an independent set in $G$ then $I = J \cap X$ is an independent set in $X$. We partition the independent sets of $G$ according to their intersection with $X$ - we then count the number of independent sets in each partition modulo $k$ and take the sum.

Let $I$ be an independent set in $X$ and let $[I]$ denote the set of independent sets in $G$ whose intersection with $X$ is $I$. Now we consider two cases.

First, assume that there is some subgraph $H_i$ such that the neighbourhood of $H_i$ does not share any vertices with $I$ (\ie such that $I\cap N_G(H_i) = \emptyset$). Now, any independent set in $[I]$ can be written as the union of an independent set in $H_i$ and an independent set in $G\backslash H_i$ the intersection of which with $X$ is $I$. Furthermore, every such union is an independent set in $[I]$. Then the total number of independent sets in $[I]$ is congruent modulo $k$ to $|\mathcal{I}(H_i)|$ multiplied by the number of independent sets in $G\backslash H_i$ whose intersection with $X$ is $I$, but since $|\mathcal{I}(H_i)|$ is congruent to zero modulo $k$, then $|[I]|$ is congruent to zero modulo $k$. Note that in this case the product term in the summation above always evaluates to zero - giving a correct count modulo $k$ of the size of $[I]$.

Now, assume that for all $i$, the neighbourhood of $H_i$ (and therefore the neighbourhood of $h_i$ does contain some vertex in $I$ ($I\cap N(H_i) \not = \emptyset$). Then any independent set in $[I]$ can be written as the union of $I$ and $n$ different independent sets $\{J_1..,J_n\}$ such that $J_i$ is entirely contained in $V(H_i\backslash\{h_i\})$ and, once again, each such union is an independent set in $[I]$. The total number of such unions is clearly $\prod_{i=1}^n |\mathcal{I}(G;H_i\backslash\{h_i\})|$, and since the minimum of $|I\cap N_G(H_i)|$ and 1 is equal to 1 for all $i$ in this case, this is equal to the product given in the theorem.
\end{proof}

The theorem \ref{ISparity} which we prove next is in fact a consequence of the theorem \ref{BiIndSet} which we prove below. However, since the reduction used here is probably easier to follow, and is similar in structure to that used in the later proof, we will give the construction of this reduction explicitly.

\begin{thm}\label{ISparity}$\bigoplus$INDEPENDENT-SET is $\bigoplus$P-complete\end{thm}

\begin{proof}

We precede by reduction from $\bigoplus$SAT. Given a CNF formula $F$ with clauses $\{C_1,\ldots C_m\}$ and variables $\{x_1,\ldots x_n\}$, considered as an instance of $\bigoplus$SAT, we construct a graph, $G$, with vertices $\{v_i, \bar{v_i}, p_i \mid i \in \{1,\ldots,n\}\}$, corresponding to each variable in $F$. There are also vertices $\{c_j \mid j = 1,\ldots,m\}$, each corresponding to one clause in $F$. There are three types of edges in the graph. Each pair $(v_i,\bar{v_i})$ is linked by an edge, and each vertex $p_i$ is linked by an edge to both $v_i$ and $\bar{v_i}$. Finally, a vertex $v_i$ ($\bar{v_i}$) is linked to a vertex $c_j$ if and only if the literal $x_i$ ($\bar{x_i}$) appears in the clause $C_j$. An example of the graph derived from the SAT formula with the single clause $x_1 \vee \bar{x_2}$ is given in figure \ref{xORybar}. We claim that the parity of the number of independent sets in $G$ is equal to the parity of the number of satisfying assignments of $F$.

This graph satisfies the conditions of lemma \ref{ISFormula}. The special subgraph $H$ is the graph on one vertex, which has $2 \equiv 0 \pmod 2$ independent sets as required. The $p_i$ and $c_j$ are the copies of $H$ and the set $X$ is the vertices $v_i,\bar{v_i}, i \in \{1...n\}$. It therefore suffices for us to show that the independent sets, $I$, of $G$ which satisfy $I\cap N_G(p_i) \neq \emptyset$ and $I\cap N_G(c_j) \neq \emptyset$ for all $i$ and $j$ are in one-to-one correspondence with the satisfying assignments of $F$.

We note that an independent set, $I$, with the required property must contain exactly one of $v_i$ and $\bar{v_i}$ for each $i$. It must contain at least one in order to ensure that $p_i$ has a neighbour in $I$, and it cannot contain more than one as $(v_i,\bar{v_i})\in E(G)$. We now consider the assignment of truth values to variables in $F$ given by setting $s(x_i)$ to true if $v_i\in I$ and setting it to false if $\bar{v_i} \in I$. To see that this assignment is satisfying, let $C_j$ be a clause in $F$, then the vertex $c_j$ has some neighbour in $I$, which is either $v_i$ or $\bar{v_i}$ for some $i$, and the literal $x_i$ or $\bar{x_i}$, which appears in $C_j$, is set to true by the construction of $s$.

Now using the formula in lemma \ref{ISFormula}, we see that the number of independent sets of $G$ modulo 2 is equal to the number of satisfying assignments of $F$ modulo 2, giving the desired reduction.

\begin{figure}[h] \label{xORybar}
\begin{center}
\begin{tikzpicture}

\path (1,2.5) node[above] (v1lab) {$v_1$};
\path (1,2.5) coordinate (v1);
\draw [fill] (v1) circle (2pt);

\path (4,2.5) node[above] (v1barlab) {$\bar{v_1}$};
\path (4,2.5) coordinate (v1bar);
\draw [fill] (v1bar) circle (2pt);

\path (2.5,1.5) node[below] (p1lab) {$p_1$};
\path (2.5,1.5) coordinate (p1);
\draw [fill] (p1) circle (2pt);

\path (5,4) node[above] (C1lab) {$C_1$};
\path (5,4) coordinate (C1);
\draw [fill] (C1) circle (2pt);

\path (9,2.5) node[above] (v2barlab) {$\bar{v_2}$};
\path (9,2.5) coordinate (v2bar);
\draw [fill] (v2bar) circle (2pt);

\path (6,2.5) node[above] (v2lab) {$v_2$};
\path (6,2.5) coordinate (v2);
\draw [fill] (v2) circle (2pt);

\path (7.5,1.5) node[below] (p2lab) {$p_2$};
\path (7.5,1.5) coordinate (p2);
\draw [fill] (p2) circle (2pt);

\draw (v1) -- (v1bar);
\draw (v1) -- (p1);
\draw (v1bar) -- (p1);
\draw (v1) -- (C1);
\draw (C1) -- (v2bar);

\draw (v2) -- (v2bar);
\draw (v2) -- (p2);
\draw (v2bar) -- (p2);
\end{tikzpicture}
\end{center}
\caption{Graph derived from formula $x_1\vee \bar{x_2}$}
\end{figure}
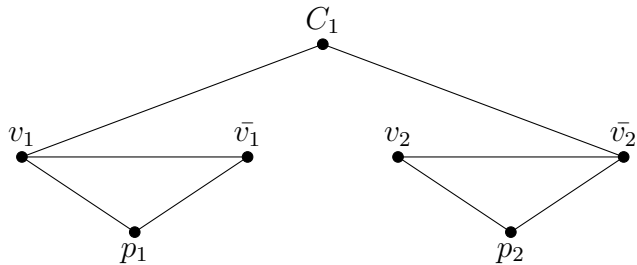
\end{proof}

\begin{thm}\label{IScountk}{\#$_k$-INDEPENDENT-SET is \#$_k$-complete for all $k$}\end{thm}

\begin{proof}
Whilst it is possible to give a construction along the lines of that given above (the special subgraphs being copies of $K_p$), this theorem is again an immediate consequence of theorem \ref{BiIndSet}, so this time we will not detail the construction explicitly.
\end{proof}

As noted above, the problem of counting independent sets in a bipartite graph is parsimoniously reducible to the problem of counting the number of satisfying assignments of an $OR_1$ formula. In order to use $OR_1$ as a base problem for our counting reduction in the next section, we prove theorem \ref{BiIndSet}.

\begin{thm} \label{BiIndSet}
\#$_k$BIPARTITE-INDEPENDENT-SET is \#$_k$-complete for all $k$
\end{thm}

\begin{proof}
We begin by noting that it actually suffices to show that the problem of counting modulo $p$ is \#$_p$-complete for all prime $p$, since counting modulo $k$ for any composite number $k$ is at least as hard as counting modulo any of the prime factors of $k$.

We proceed by reduction from \#$_p$-SAT. Given a SAT formula $F$ with clauses $\{C_1,\ldots C_m\}$ and variables $\{x_1,\ldots x_n\}$ and a prime number $p$, we construct a graph as described below.

With each variable $x_i$ in $F$, we associate a subgraph of $G$ as follows, the subgraph contains special vertices $v_i$, $\bar{v_i}$ the presence or absence of which in an independent set will correspond to the truth or otherwise of the literals $x_i$,$\bar(x_i)$ of $F$, there are also vertices $p_i$ and $\bar{p_i}$ - these are connected to $v_i$ and $\bar{v_i}$ respectively, and are both connected by an edge to one vertex, $h_i$ in $H_i$. Where $H_i$ is a copy of $H$, a bipartite graph with a distinguished vertex $h$, having the property that the number of independent sets in $H$ is a multiple of $p$ and that the number of independent sets in $H\backslash\{h\}$ is non-zero modulo $p$. There is also another copy of the same graph, $H_i^*$, one vertex of which, $h_i^*$ is linked by an edge to each of $v_i$ and $\bar{v_i}$. Finally for each clause $C_j$ in $F$ we add another copy of this bipartite graph $H$, denoted $C_j$, one vertex of which, $c_j$ is linked to each of the vertices representing the literals present in the clause $C_j$.

Formally then, the vertex set of $G$ will be $\{v_i,\bar{v_i},p_i,\bar{p_i}\mid i = 1 \ldots n\}$. Along with $\{V(H_i),V(H_i)^* \mid i = 1\ldots n\}$ and $\{V(C_j) \mid j = 1 \ldots m\}$ copies of $H$. The edge set will be the edges of $H_i$, $H_i^*$ and $C_j$ along with the edges $\{(v_i,p_i),(\bar{v_i},\bar{p_i}),(p_i,h_i),(\bar{p_i},h_i),(v_i,h_i^*),(\bar{v_i},h_i) \mid i = 1 \ldots n\}$ and the edges $(v_i,c_j)$, $(\bar{v_i},c_j)$ such that the literals $x_i$, $\bar{x_i}$ respectively appear in the clause $c_j$.

An example of the subgraph associated with a variable $x_i$ lying in the clause $C_j$ is given in figure \ref{Hfig}.

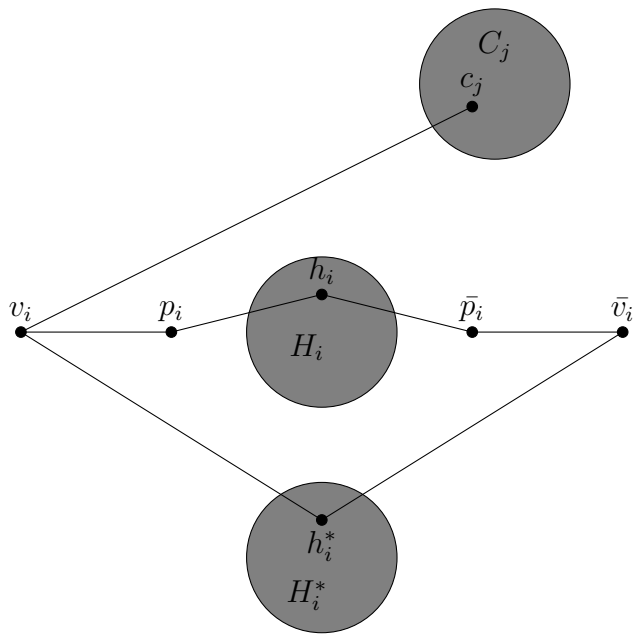
\begin{figure}[h] \label{Hfig}
\begin{center}
\begin{tikzpicture}

\draw [fill = gray] (5,4) ellipse (1 and 1);
\path (4.8,3.8) node (Hi) {$H_i$};
\draw [fill = gray] (7.3,7.3) ellipse (1 and 1);
\path (7.3,7.8) node (Cj) {$C_j$};
\draw [fill = gray] (5,1) ellipse (1 and 1);
\path (4.8,0.5) node (Histar) {$H_i^*$};

\path (1,4) node[above] (vilab) {$v_i$};
\path (1,4) coordinate (vi);
\draw [fill] (vi) circle (2pt);

\path (9,4) node[above] (vibarlab) {$\bar{v_i}$};
\path (9,4) coordinate (vibar);
\draw [fill] (vibar) circle (2pt);

\path (3,4) node[above] (pilab) {$p_i$};
\path (3,4) coordinate (pi);
\draw [fill] (pi) circle (2pt);

\path (7,4)   node[above] (pibarlab) {$\bar{p_i}$};
\path (7,4) coordinate (pibar);
\draw [fill] (pibar) circle (2pt);

\path (5,4.5) node[above] (hilab) {$h_i$};
\path (5,4.5) coordinate (hi);
\draw [fill] (hi) circle (2pt);

\path (5,1.5) node[below] (histarlab) {$h_i^*$};
\path (5,1.5) coordinate (histar);
\draw [fill] (histar) circle (2pt);

\path (7,7) node[above] (cjlab) {$c_j$};
\path (7,7) coordinate (cj);
\draw [fill] (cj) circle (2pt);
\draw (cj) -- (vi) -- (pi) -- (hi) -- (pibar) -- (vibar) -- (histar) -- (vi);

\end{tikzpicture}
\end{center}
\caption{Subgraph associated with the variable $x_i$}
\end{figure}

Using lemma \ref{ISFormula} it suffices to show that the independent sets, $I$, of $G$ which satisfy $I\cap N(H_i) \not = \emptyset$ for all $i$ are in one-to-one correspondence with the satisfying assignments of $F$, and that we can produce a bipartite graph $H$ with the desired property. Since we are able to divide by any non-zero constant modulo $p$, the formula given in the lemma will then allow us to derive the number of satisfying assignments of $F$ modulo $p$ from the number of independent sets of $G$ modulo $p$, giving the required reduction. That such $H$ can be constructed for all prime $p$ is shown in lemma \ref{Hexists} below.

Let $I$ be an independent set in $G$ with the relevant property. Then for all $i$, either $\{v_i,\bar{p_i}\}\subset I$ or $\{\bar{v_i},p_i\} \subset I$. To see this, we note that both $h_i$ and $h_i^*$ have some neighbour in $I$ by assumption, but then the only neighbours of $h_i$ in $X$ are $v_i$ and $\bar{v_i}$, so one of these two must be in $I$. Similarly, the neighbours of $h_i^*$ in $G$ are $p_i$ and $\bar{p_i}$ - so one of this pair must be in $I$, but then since $I$ is independent and $(v_i,p_i),(\bar{v_i},\bar{p_i})\in E(G)$ we have the stated result. Let $s$ be the assignment of truth values to variables in $F$ given by $s(x_i) = \textrm{true} \iff v_i \in I$. We claim that this is a satisfying assignment of $F$.

Indeed, let $C_j$ be a clause of $F$. Then there is some element of $I$ which is a neighbour of $c_j$, the distinguished node in $C_j$. This is either $v_i$ or $\bar{v_i}$ for some $i$, but then the literal $x_i$ ($\bar{x_i}$) appears in the clause $C_j$, and this literal is true by construction of $s$, therefore the clause $C_j$ is satisfied.

Similarly, if $s$ is a satisfying assignment of $F$, then the independent set constructed analogously to that above (with $\{v_i,\bar{p_i}\}\subset I$ if $s(x_i) =$ true and $\{\bar{v_i},p_i\} \subset I$ otherwise) is an independent set of $X$ with the required property.
\end{proof}

\begin{lem}\label{Hexists}
For all prime $p$ it is possible to construct a bipartite graph $H$, containing a distinguished node $h$, with the following properties.
\begin{itemize}
          \item[(i)] The number of independent sets in $H$ is congruent to zero modulo $p$.
          \item[(ii)] The number of independent sets in $H\backslash \{h\}$ is not congruent to zero modulo $p$.
\end{itemize}
Furthermore, for $p>2$, the graph $K_{(p-2),(p-2)}$ is such an $H$ (any node of the graph can be chosen as the distinguished node, since they are indistinguishable).
\end{lem}
\begin{ex} \label{H2}
An example of a subgraph H which would satisfy the above conditions for $p = 2$ is the graph on one vertex, where the distinguished vertex, $h$ will clearly be the unique vertex in the graph. This graph has precisely 2 independent sets ($\emptyset$ and $\{h\}$), whereas $H\backslash \{h\} = \emptyset$ has precisely one.
\end{ex}
\begin{proof} [Proof of Lemma \ref{Hexists}]
We note that the graph $K_1$ provides an example of such a graph for $p=2$ (as explained in example \ref{H2}), and therefore restrict our attention to the case $p>2$.

Consider the graph $K_{n,n}$, the complete bipartite with two classes of $n$ vertices each. This graph has $2^{(n+1)} -1$ independent sets. To see this, simply note that any independent set in $K_{n,n}$ is contained entirely in one of the two vertex classes, and that every subset of one of the vertex classes in independent. Then there are $2^n$ independent sets in each class, but the empty set is in both, so there are in fact $2^{n+1}-1$ independent sets in $K_{n,n}$.

Now let $n=p-2$, then $K_{n,n}$ has $2^{p-1} - 1$ independent sets. But by Fermat's little theorem, $2^{p-1} \equiv 1 \pmod p$, therefore the number of independent sets of $K_{p-2,p-2}$ is congruent to zero modulo $p$.

Finally, the number of independent sets in $K_{p-2,p-3}$ (which is $H$ with a vertex deleted) is equal to $2^{p-2}+2^{p-3} -1$, but this is just $(2^{p-1} - 1) - 2^{p-3}$, and since $2 \not \equiv 0 \pmod p$, we have that $2^{p-3} \not \equiv 0 \pmod p$, and so $(2^{p-1} - 1) - 2^{p-3} \equiv - 2^{p-3} \not \equiv 0 \pmod p$. \end{proof}

\begin{proof}[Proof of theorem \ref{ORComplexity}]
This follows immediately from the theorems \ref{IScountk} and \ref{BiIndSet} along with lemmas \ref{Mon2SATisIS} and \ref{OR1BiIS}.
\end{proof}

\section{The classes \#$_k$-SAT} \label{SectionSAT}
We now know that  \#$_k$SAT$(OR_0)$, \#$_k$SAT$(OR_1)$ and \#$_k$SAT$(OR_2)$ are \#$_k$P-complete for all integer $k$. We proceed to give reductions from these base problems to Generalised Satisfiability Problems - the reductions are in most cases identical to those used by Creignou et. al. in \cite{CreKha01}.

We will make use of the functions T and F which are the functions of one variable which evaluate to true and false respectively, as well as \textbf{$XOR(x,y)$}, the function which evaluates to true when exactly one of $\textbf{x}$ and $\textbf{y}$ is true and false otherwise. We will say that a constraint set $\mathcal{F}$ is \textbf{C-closed} if every constraint function $f$ in $\mathcal{F}$ is such that if $f(\textbf{x})$ is true then $f(\textbf{1-x})$ is also true, in other words, such that the set of satisfying assignments of an $\mathcal{F}$-constraint is closed under complement. We will say that a constraint set, $\mathcal{F}$, is \textbf{0-valid (1-valid)} if setting all of the variables in any $\mathcal{F}$-formula to 0 (1) results in the formula evaluating to true. Finally, a constraint set is \textbf{affine} if each of the constraints in the set can be expressed as a system of linear equations in GF2.

\begin{defn}\label{implDef}
A family of constraints, $\mathcal{F}$, over a set of variables $\textbf{x}$, $\textbf{y}$, faithfully implements a boolean function $f(x)$ iff there exists an $\mathcal{F}$-collection of constraints, $\mathcal{C}$ such that there is exactly one way to satisfy each constraint in $\mathcal{C}$ whenever $f(\textbf{x})$ evaluates to true, and no ways to satisfy them all whenever $f(\textbf{x})$ evaluates to false. The variables $\textbf{x}$ are called function variables, and the variables $\textbf{y}$ auxiliary variables.
\end{defn}

We note that for our purposes a slightly weaker definition of faithful implementation would suffice, with ``exactly one" replaced with ``exactly one modulo $k$". However, it turns out that the reductions we need are faithful in the original sense, and therefore we use this definition in order to be able to appeal directly to the results of \cite{CreKha01}.

\begin{ex}
The constraint family $\{OR_0,\mathrm{F}\}$ faithfully implements the function T$(\textbf{x})$ through the constraint applications $\{OR_0(\textbf{x},\textbf{y}),\bar{\textbf{y}}\}$, $\textbf{y}$ is an auxiliary variable.
\end{ex}

\begin{lem}\label{implementations}
Given an integer $k$ and a constraint set $\mathcal{F}$, if \#$_k$SAT($\mathcal{F}$) is \#$_k$P-hard and every constraint of $\mathcal{F}$ can be faithfully implemented by $\mathcal{F'}$, then \#$_k$SAT($\mathcal{F'}$) is also \#$_k$P-hard.
\end{lem}

\begin{proof}
This proof is essentially identical to the proof of theorem 5.15 in \cite{CreKha01}. Given an $\mathcal{F}$-collection of constraint applications on a variable set $\textbf{x}$, say $\mathcal{C}$, we transform this using faithful implementations to an $\mathcal{F'}$-collection of constraint applications on a new variable set, $(\textbf{x},\textbf{y})$, say $\mathcal{C'}$. Since the implementations are faithful, each satisfying assignment of $\mathcal{C}$ can be extended in a unique way to a satisfying assignment of $\mathcal{C'}$. Therefore there is a one-to-one correspondence between satisfying assignments of $\mathcal{C}$ and satisfying assignments of $\mathcal{C'}$. This gives a parsimonious reduction from \#SAT($\mathcal{F}$) to \#SAT($\mathcal{F'}$), which clearly implies the desired result.
\end{proof}

We will make use of the following lemmas, taken from \cite{CreKha01} and stated here without proof.

\begin{lem}\label{not01val}
\cite{CreKha01} If a constraint family $\mathcal{F}$ is not 0-valid (1-valid) and
\begin{itemize}
 \item[(i)] if $\mathcal{F}$ is C-closed, then $\mathcal{F}$ faithfully implements XOR.
 \item[(ii)] if $\mathcal{F}$ is not C-closed, then $\mathcal{F}$ faithfully implements T (F).
\end{itemize}
\end{lem}

\begin{lem}\label{withConst}
\cite{CreKha01} Take a function $f$. If $f$ is not affine, then $\{f,\mathrm{F,T}\}$ faithfully implements at least one of the three functions OR$_0$, OR$_1$ and OR$_2$. Furthermore, if $f$ is 0-valid (1-valid) then $\{f,\mathrm{F}\}$ ($\{f,\mathrm{T}\}$) faithfully implements one of OR$_1$ or OR$_2$ (OR$_0$ or OR$_1$).
\end{lem}

\begin{lem}\label{nonCClosed}
Let $\mathcal{F}$ be a non-C-Closed family of functions. Then if $\mathcal{F}$ is both 0-valid and 1-valid, $\mathcal{F}$ faithfully implements OR$_1$.
\end{lem}

We also need the following lemmas, which have been adapted from the versions given in \cite{CreKha01}.

\begin{lem} \label{CClosed}
Let $\mathcal{F}$ be a set of C-Closed functions. If $p$ is an odd prime, and if \#$_p$SAT($\mathcal{F} \cup \{\mathrm{F,T}\}$) is \#$_p$P-hard and if $\mathcal{F}$ can faithfully implement the XOR function, then \#$_p$SAT($\mathcal{F}$) is \#$_p$P-hard.
\end{lem}
\begin{proof}
We will use the following reduction: Let $\mathcal{C}$ be an $\mathcal{F} \cup \{\mathrm{F,T}\}$-collection of constraint applications on variables $\textbf{x}$ let $y_0$, $y_1$ be two new variables, and replace with $y_0$ any variable constrained to be false, and replace with $y_1$ any variable constrained to be true. Now add the constraint $XOR(y_0,y_1)$. We now have, $\mathcal{C'}$ an $\mathcal{F} \cup XOR$ collection of constraint applications on variables $\textbf{x},y_0,y_1$. Clearly any satisfying assignment of $\mathcal{C}$ can be extended to a satisfying assignment of $\mathcal{C'}$ by setting $s'(y_0)=0$ and $s'(y_1)=1$. Conversely, let $s'$ be a satisfying assignment of $\mathcal{F'}$ then either $s'(y_0)=0$ and $s'(y_1)=1$, in which case $s'$ restricted to $\textbf{x}$ is a satisfying assignment of $\mathcal{C}$ or $s'(y_0)=1$ and $s'(y_1)=0$, in which case it is easy to check that $s(x) = 1-s'(x)$ satisfies all constraints in $\mathcal{C}$. So $\mathcal{C'}$ has precisely twice as many satisfying assignments as $\mathcal{C}$.

Now since p is prime and $p\geq 2$, we can divide by two modulo $p$, giving a Turing reduction from \#$_p$SAT($\mathcal{F} \cup \{\mathrm{F,T}\})$ to \#$_p$SAT($\mathcal{F}$. Finally, since $\mathcal{F}$ can faithfully implement XOR, we have \#$_p$P-hardness of \#$_p$SAT($\mathcal{F}$) by lemma \ref{implementations}.
\end{proof}

\begin{lem} \label{CClosedPowersof2}
Let $\mathcal{F}$ be a set of C-Closed functions. For all integer $k$, if \#$_{2^{k-1}}$SAT($\mathcal{F} \cup \{\mathrm{F,T}\}$) is \#$_{2^{k-1}}$P-hard and if $\mathcal{F}$ can faithfully implement the XOR function then \#$_{2^k}$SAT($\mathcal{F}$) is \#$_{2^k}$P-hard.
\end{lem}
\begin{proof}
The reduction used is the same as in the previous proof. Now, given a $\mathcal{F} \cup \{\mathrm{F,T}\}$-formula, $F$, we have constructed a $\mathcal{F}$ formula, $F'$ with twice as many satisfying assignments as $F$. Now, any algorithm which could count the number of solutions of $F'$ modulo $2^k$ in polynomial time could clearly be used to count the number of solutions of $F$ modulo $2^{k-1}$ in polynomial time. Therefore \#$_{2^{k-1}}$P-hardness of \#$_{2^{k-1}}$SAT($\mathcal{F} \cup \{\mathrm{F,T}\}$) implies \#$_{2^k}$P-hardness of \#$_{2^k}$SAT($\mathcal{F}$) as required.
\end{proof}

\begin{lem} \label{CClosedFalse}
Let $\mathcal{F}$ be a set of C-Closed functions. If $p$ is an odd prime, and if \#$_p$SAT($\mathcal{F} \cup \{$F$\}$) is \#$_p$P-hard and if $\mathcal{F}$ then \#$_p$SAT($\mathcal{F}$) is \#$_p$P-hard.
\end{lem}
\begin{proof}
We construct a $\mathcal{F}$ formula from a given $\mathcal{F} \cup \{ \mathrm{F}\}$ formula by replacing all variables which are constrained to be false with a new variable $x_0$. This formula then has twice as many satisfying assignments as the original, and we proceed as in the proof of lemma \ref{CClosed}.
\end{proof}

\begin{lem} \label{CClosedFalsePowersof2}
Let $\mathcal{F}$ be a set of C-Closed functions. For all integer $k$, if \#$_{2^{k-1}}$SAT($\mathcal{F} \cup \{$F$\}$) is \#$_{2^{k-1}}$P-hard  then \#$_{2^k}$SAT($\mathcal{F}$) is \#$_{2^k}$P-hard.
\end{lem}
\begin{proof}
Using the same reduction as in the proof of the previous lemma, and then the same reasoning as in the proof of lemma \ref{CClosedPowersof2} we obtain the desired result.\end{proof}

Finally, we require the observation that for C-Closed functions, the number of satisfying assignments modulo 2 is always equal to zero - as for any satisfying assignment $\textbf{s}$, the assignment \textbf{$\textbf{1-s}$} is also satisfying.

\begin{thm}\label{countCSPmodk}
Given a constraint set $\mathcal{F}$, and an integer k, the problem \#$_k$SAT($\mathcal{F}$) is in FP if $\mathcal{F}$ is an affine family of constraints, or if $k = 2$ and $\mathcal{F}$ is C-closed, and it is otherwise $\#_k$P-complete.
\end{thm}

\begin{proof}
There are several cases to consider, first we note that \#$_k$SAT($\mathcal{F}$) is clearly in \#$_k$P. Now, if every constraint in $\mathcal{F}$ is affine, then we can consider solving \#SAT($\mathcal{F}$) as the problem of solving a system of linear equations of GF(2), this can be done using Gaussian elimination in polynomial time. Since we can solve \#SAT($\mathcal{F}$) in polynomial time, we can clearly solve \#$_k$SAT($\mathcal{F}$) in polynomial time. Also, if  $\mathcal{F}$ is C-closed, then clearly $\mathcal{F}$ has an even number of satisfying assignments, so the problem \#$_2$SAT($\mathcal{F}$) is trivial, and can certainly be solved in polynomial time.

Now, suppose $\mathcal{F}$ contains a function, $g$, which is not affine, and that if $k=2$ then $\mathcal{F}$ is not C-closed. There are 3 cases.

\begin{description}
  \item[g is neither 0-valid nor 1-valid] Then family $\{g,$F,T$\}$ can faithfully implement one of OR$_0$, OR$_1$ and OR$_2$ (Lemma \ref{withConst}). Hence by lemma \ref{implementations} and theorem \ref{ORComplexity}, \#$_k$SAT($\mathcal{F} \cup \{$F,T$\}$) is \#$_k$P-complete for all $k$. If $\mathcal{F}$ contains a function which is not C-closed, we can faithfully implement F and T by lemma \ref{not01val} so we get \#$_k$-hardness for \#$_k$SAT($\mathcal{F}$). Otherwise we can faithfully implement XOR by lemma \ref{not01val} and we get \#$_p$-hardness for all odd primes $p$ using lemma \ref{CClosed}, and \#$_{2^l}$-hardness for all $l \geq 2$ using lemma \ref{CClosedPowersof2}. Now we have \#$_k$-hardness for all $k \geq 3$ (as all $k\geq3$ have as a factor either some odd prime or some power of two greater than or equal to four, and counting modulo $k$ is at least as hard as counting modulo any factor of $k$).

  \item[g is 0-valid but not 1-valid (or vice versa)] In this case, $\{g,$F$\}$ can faithfully implement one of the functions OR$_1$ or OR$_2$ (lemma \ref{withConst}). Also, clearly $g$ itself can faithfully implement F since it is 0-valid but not 1-valid. Thus $\mathcal{F}$ can faithfully implement one of OR$_1$ or OR$_2$. Then by the lemma \ref{implementations} and theorem \ref{ORComplexity}, we get \#$_k$-hardness for \#$_k$SAT($\mathcal{F}$).
      Note that in this case $g$ itself is not C-Closed as $g(0)$ = true and $g(1)$ = false so we don't need to deal with the possibility that $\mathcal{F}$ is C-Closed.

  \item[g is 0-valid and 1-valid] Then if $g$ is not C-closed, $g$ can faithfully implement OR$_1$ (lemma \ref{nonCClosed}) which gives \#$_k$-hardness for \#$_k$SAT($\mathcal{F}$).
      Otherwise $\{g,F\}$ can faithfully implement one of OR$_1$ or OR$_2$ (lemma \ref{withConst}), which gives \#$_k$-hardness of \#$_k$SAT($\mathcal{F},F$). Therefore we can use lemmas \ref{CClosedFalse} and \ref{CClosedFalsePowersof2} to get \#$_k$-hardness of \#$_k$SAT$(\mathcal{F})$.
\end{description}
\end{proof}

\bibliographystyle{plain}
\bibliography{CSPbib}

\end{document}